\documentclass[11pt, a4paper]{article}
\usepackage[textwidth = 7in, textheight=25cm,nohead]{geometry}
\usepackage{ulem}
\usepackage{arydshln}
\usepackage{graphicx, amssymb, amsmath, amsthm, amstext, booktabs, float, multirow, subfigure, rotating,natbib,bm}
\usepackage{tikz}
\usepackage{setspace}
\usepackage{enumerate}

\numberwithin{equation}{section} 
\numberwithin{figure}{section} 
\numberwithin{table}{section}

\newtheorem{proposition}{Proposition}[section]

\theoremstyle{remark} 
\newtheorem{remark}{Remark}[section] 

\theoremstyle{definition}

\newcommand\indep{\protect\mathpalette{\protect\independenT}{\perp}}
\def\independenT#1#2{\mathrel{\rlap{$#1#2$}\mkern2mu{#1#2}}}

\bibpunct{(}{)}{,}{a}{,}{,}

\def\ppn{\vskip 6pt \noindent }
\def\R{{\mathbb{R}}}
\def\N{{\mathbb{N}}}

\newcommand{{\Xs}}{{\cal X}}
\newcommand{{\Ys}}{{\cal Y}}
\newcommand{{\Ls}}{{\cal L}}
\newcommand{{\Ss}}{{\cal S}}
\newcommand{{\Ms}}{{\cal M}}
\newcommand{{\Gs}}{{\cal G}}
\newcommand{{\As}}{{\cal A}}
\newcommand{{\Hs}}{{\cal H}}
\newcommand{{\Ns}}{{\cal N}}
\newcommand{{\Is}}{{\cal I}}
\newcommand{{\Vs}}{{\cal V}}
\newcommand{{\Ds}}{{\cal D}}
\newcommand{{\Bs}}{{\cal B}}
\newcommand{{\Cs}}{{\cal C}}
\newcommand{{\Rs}}{{\cal R}}
\newcommand{{\Fs}}{{\cal F}}
\newcommand{{\Us}}{{\cal U}}
\newcommand{{\ttheta}}{{\bm{\theta}}}
\newcommand{{\Ttheta}}{{\bm{\Theta}}}
\newcommand{{\Sss}}{{\bm{\Ss}}}
\newcommand{{\XXi}}{{\bm{\Xi}}}
\newcommand{{\pp}}{{\mathbf p}}
\newcommand{{\uu}}{{\mathbf u}}
\newcommand{{\vv}}{{\mathbf v}}
\newcommand{{\sss}}{{\mathbf s}}
\newcommand{{\ttt}}{{\mathbf t}}

\newcommand{{\FF}}{{\mathbf F}}
\newcommand{{\XX}}{{\mathbf X}}
\newcommand{{\UU}}{{\mathbf U}}
\newcommand{{\YY}}{{\mathbf Y}}
\newcommand{{\TT}}{{\mathbf T}}
\newcommand{{\VV}}{{\mathbf V}}
\newcommand{{\BB}}{{\mathbf B}}
\newcommand{{\KK}}{{\mathbf K}}
\newcommand{{\HH}}{{\mathbf H}}
\newcommand{{\II}}{{\mathbf I}}
\newcommand{{\xx}}{{\mathbf x}}
\newcommand{{\yy}}{{\mathbf y}}
\newcommand{{\bb}}{{\mathbf b}}
\newcommand{{\ab}}{{\mathbf a}}

\newcommand{{\toL}}{{\overset{\mathcal{L}}{\longrightarrow}\ }}

\newcommand{{\MC}}{{\,  *_{\text{\scalebox{0.65}{$\Ms$}}}\,  }}

\newcommand{{\dou}}{$\leadsto$\ }

\DeclareMathOperator{\Supp}{Supp}

\DeclareMathOperator{\dCor}{dCor}

\newcommand{\indic}[1]{
	\hbox{${\it 1}\hskip -4.5pt I_{\{ #1 \}}$}
}

\begin{document}
\title{The Hellinger correlation}
\author{\sc{Gery Geenens}\thanks{Corresponding author: {\tt ggeenens@unsw.edu.au}, School of Mathematics and Statistics, UNSW Australia, Sydney, tel +61 2 938 57032, fax +61 2 9385 7123 } \quad \sc{Pierre Lafaye de Micheaux}\\ School of Mathematics and Statistics, UNSW Sydney, Australia
}
\date{\today}
\maketitle
\thispagestyle{empty}

\begin{abstract}
  In this paper, the defining properties of any valid measure of the dependence between two continuous random variables are revisited and complemented with two original ones, shown to imply other usual postulates. While other popular choices are proved to violate some of these requirements, a class of dependence measures satisfying all of them is identified. One particular measure, that we call the Hellinger correlation, appears as a natural choice within that class due to both its theoretical and intuitive appeal. A simple and efficient nonparametric estimator for that quantity is proposed, with its implementation publicly available in an \texttt{R} package. Synthetic and real-data examples illustrate the descriptive ability of the measure. %, which can also be used as test statistic for exact independence testing.
\end{abstract}

\section{Introduction} \label{sec:intro}

A large part of science is about understanding the dependence between several factors that may influence each other, for instance to disentangle genetics and environmental risk factors for individual diseases. Hence statistics, the art of turning empirical evidence into information, has always kept dependence modelling at its core. Characterising the dependence between two variables includes two main tasks: $(i)$ testing the null hypothesis $H_0$ that the two variables are independent; and $(ii)$ measuring the strength of any dependence that may exist between the two variables. %This paper focuses mainly on the latter.

\ppn These two tasks are often amalgamated, for instance when a dependence measure is used as the test statistic in an independence testing procedure, or when the observed value of the test statistic is interpreted as a quantification of the underlying dependence. This may seem natural, however there are good reasons to approach them separately. Indeed a measure is expected to be a fair quantification of the strength of the involved dependence. In a case of weak but non-null dependence, we would expect a reliable measure of dependence to take a value accordingly low. By contrast, independence testing aims at making a binary decision as to the presence of dependence. As such, a powerful test should be based on a statistic that takes values as large as possible (i.e., very different than under $H_0$) as soon as dependence is present, regardless of its strength. Consequently, using an interpretable dependence measure as a test statistic typically implies a loss of power for the resulting test, while a test statistic designed for guaranteeing high power for the resulting test, typically lacks a bit of finesse for accurately quantifying the strength of the dependence \citep{Reimherr13,Sun14}. This paper focuses on meaningfully measuring dependence, without explicitly giving a central role to concepts inspired by dichotomous testing procedures, such as power.

\ppn The literature on quantifying dependence has long been monopolised by the historical measures, such as Pearson's correlation, Spearman's $\rho$, Kendall's $\tau$ and Hoeffding's $D$. Yet, the interest in modern alternatives has recently made an impressive upsurge. Among others, one can cite distance correlation \citep{Szekely07}, maximal information criterion \citep{Reshef11}, %`HHG' \citep{Heller13}, 
Hilbert-Schmidt independence criterion \citep{Gretton05,Pfister18}, sign covariance \citep{Bergsma14}, G-squared \citep{Wang17} and symmetric rank covariance \citep{Weihs18}, along with a renewed enthusiasm for the mutual information \citep{Kinney14,Zeng18,Berrett19}. 

\ppn Through this abundance resurfaces the question of the criteria discriminating sensible measures from others. Far from new, it was already addressed by \cite{Renyi59}, who formulated 7 axioms that a valid dependence measure between two variables, say $X_1$ and $X_2$, should satisfy. However, the only measure known to fulfil them all is the {\it maximal correlation coefficient}, viz.\ $S(X_1,X_2) = \sup_{\psi_1,\psi_2} |\rho\left(\psi_1(X_1),\psi_2(X_2)\right)|$, where $\rho(\cdot,\cdot)$ is Pearson's correlation and the supremum is taken over all Borel functions $\psi_1,\psi_2: \R \to \R$. Yet, that measure is computable only in very particular cases, and it may return, and often does so, its maximal value $S(X_1,X_2) = 1$ even in the absence of any obvious strong dependence between $X_1$ and $X_2$ \citep{Bell62}. This evidences that some of Renyi's axioms may be unsuitable for general use. As a result, different sets of amended axioms have been proposed in the subsequent literature, see e.g., \citet{Schweizer81}, \cite{Lancaster82}, \cite{Granger04} or \citet[Section 4.3]{Balakrishnan09}. Among those propositions, five properties, labelled (P\ref{ax:exist})--(P\ref{ax:Gaussconf}) hereafter, seem difficult to dispute, while others are more prone to discussion. Below, we complement those 5 mainstream postulates with two original ones (P\ref{ax:funct})--(P\ref{ax:GDPI}), and justify at length their reason-of-being. We show that they are more fundamental than other usually posited properties, while more natural intuitively. %We point out how some popular dependence measures violate some of the suggested 7 postulates, while we identify a class of dependence measures which satisfy them. Finally, we develop on one such measure in particular, leading to the definition of the {\it Hellinger correlation}, for which we also propose a simple and efficient estimator.

\section{Renyi's axioms and beyond} \label{sec:dep}

\subsection{Dependence} \label{sec:dependence}

We call {\it dependence} between two random variables $X_1$ and $X_2$ whatever remains to be specified for entirely reconstructing the joint distribution $F_{12}$ of $(X_1,X_2)$ once their marginal distributions $F_1$ and $F_2$ are known. The strength of dependence is thus the size (in an appropriate sense) of that missing link. As such, $X_1$ and $X_2$ are as dependent as can be when $F_{12}$ is as different as can be to the independence base case $F_1F_2$. If both $X_1$ and $X_2$ are continuous, this is characterised by $F_{12}$ being singular with respect to $F_1 F_2$ \citep{Silvey64}. In the discrete case, though, such singularity is impossible. %, and a sense of maximum dependence should be apprehended differently. 
This illustrates why measuring dependence may be a structurally different problem in the continuous and in the discrete cases \citep{Hoeffding42}. In particular, it is known that approaches based on copulas, warranted in the continuous setting, are doomed to failure for analysing dependence between discrete variables \citep{Genest07}. This justifies to study the two situations separately; this paper considers the continuous case only. Perspectives of extension of the below discussion to the discrete setting are provided in Section \ref{sec:persp}.

\subsection{Postulates}  \label{subsec:postulates}

Let $X_1$ and $X_2$ be two continuous random variables defined on the same probability space. It is widely accepted that a valid measure $D$ of the dependence between them should be such that:
\begin{enumerate}[({P}1)]
 \item \label{ax:exist} ({\it existence}) $D(X_1,X_2)$ is defined, whatever the variables $X_1$ and $X_2$;
 \item \label{ax:sym} ({\it symmetry}) $D(X_1,X_2) = D(X_2,X_1)$;
 \item \label{ax:norm} ({\it normalisation}) $0 \leq D(X_1,X_2) \leq 1$;
 \item \label{ax:indep} ({\it characterisation of independence}) $D(X_1,X_2)=0$ $\iff$ $X_1$ and $X_2$ are independent ($X_1 \indep X_2$);
 \item \label{ax:Gaussconf} ({\it weak Gaussian conformity}) If $(X_1,X_2)$ is a bivariate Gaussian vector, then $D(X_1,X_2)$ is a strictly increasing function of $|\rho(X_1,X_2)|$.
\end{enumerate}

`Existence' (P\ref{ax:exist}) is a minimal requirement. Though, many popular measures do not satisfy it. For instance, Pearson's correlation and distance correlation ($\dCor$) are only defined if $X_1$ and $X_2$ have finite second moment. 

\ppn Defined as the void between $F_{12}$ and $F_1F_2$ (Section~\ref{sec:dependence}), the dependence in $(X_1,X_2)$ is evidently invariant to permutation of $X_1$ and $X_2$, making `symmetry' (P\ref{ax:sym}) unquestionable as well. Note that asymmetric measures, arguments in favour of which may sometimes be found in the literature, explicitly target directional notions such as `regression dependence' \citep{Dette13} or causal relationships \citep{Janzing13}, hence are of a different nature and do not measure dependence as defined hereinabove.

\ppn `Normalisation' (P\ref{ax:norm}) only aims to provide benchmarks -- any candidate measure can be made comply with it through renormalisation. Importantly, though, it implies that $D(X_1,X_2)$ is an unsigned number. Any signed measure, whose sign is meant to be informative about the `direction' of the association ($D(X_1,X_2)>0$: positive association, $X_1$ and $X_2$ tend to vary in the same direction; $D(X_1,X_2)<0$: negative association, $X_1$ and $X_2$ tend to vary in opposite direction) is meaningful only when explicitly looking for such monotonic association. Dependence, as defined in Section~\ref{sec:dependence}, is a much broader concept and cannot generally be categorised as `positive' or `negative'. For instance, among the scatter-plots shown in Figure~\ref{fig:HHGplots} below, none (but [7]) exhibits any sense of `positive' or `negative' association between $X_1$ and $X_2$ while all (but [6]) involve dependence between them. Hence a general dependence measure must be unsigned. Then it seems only natural to ask the measure to be null in the case and only in the case of no dependence (P\ref{ax:indep}). 

\ppn Finally, `weak Gaussian conformity' %\footnote{\cite{Renyi59}'s original Axiom G) asked $D(X_1,X_2) = |\rho(X_1,X_2)|$ in the bivariate Gaussian case, here referred to as `{\it strong Gaussian conformity}'.} 
(P\ref{ax:Gaussconf}) is unavoidable. In a bivariate Gaussian vector, Pearson's correlation $\rho$ uniquely specifies the joint distribution once the marginals are fixed. Hence dependence (as defined in Section~\ref{sec:dependence}) is unequivocally characterised by Pearson's correlation, and any measure disagreeing with it in a bivariate Gaussian vector cannot be valid. 

\ppn Mostly dictated by common sense, these (P\ref{ax:exist})--(P\ref{ax:Gaussconf}) can be found under this form or slightly amended in most of the related references. Here, we complete this list with the following two original requirements which, by contrast, offer novel perspectives on what characterises valid dependence measures.

\begin{enumerate}[({P}1)]
	\setcounter{enumi}{5}
	%\item \label{ax:funct} ({\it characterisation of perfect dependence}) $D(X_1,X_2)=1 \iff$ there exists a Borel function $\Psi: [0,1] \to \R^2$, with image of null 2-dimensional Lebesgue measure, such that $(X_1,X_2) = \Psi(U)$ for $U \sim \Us_{[0,1]}$; 
	%\item \label{ax:funct} ({\it characterisation of pure dependence}) $D(X_1,X_2)=1$ if and only if there exists a Borel function $\Psi: [0,1] \to \R^2$ such that $(X_1,X_2) = \Psi(U)$ for $U \sim \Us_{[0,1]}$ and the image of $\Psi$ is such that $\iota(\Psi([0,1])) = 0$, where $\iota$ is the probability measure induced by $F_1F_2$ on $\R^2$;
	\item \label{ax:funct} ({\it characterisation of pure dependence}) $D(X_1,X_2)=1$ if and only if there exists a Borel function $\Psi: [0,1] \to \R^2$ such that $(X_1,X_2) = \Psi(U)$ for $U \sim \Us_{[0,1]}$ and $\Cs_\Psi$, the image of $\Psi$ in $\R^2$, is such that $\iint_{\Cs_\Psi} dF_1(x_1)\, dF_2(x_2)=0$;
	\item \label{ax:GDPI} ({\it generalised Data Processing Inequality}) If $X_1$ and $X_3$ are conditionally independent given $X_2$ ($X_1 \indep X_3 \mid X_2$), then $D(X_1,X_3) \leq \min\{D(X_1,X_2),D(X_2,X_3)\}$.
\end{enumerate}

Analogously to (P\ref{ax:indep}), $D(X_1,X_2)$ should be maximum if and only if there exists some sort of perfect dependence between $X_1$ and $X_2$. Yet, a universally accepted definition of what perfect dependence is, has proved elusive. Our interpretation of it, which we refer to as {\it pure dependence} and leads to (P\ref{ax:funct}), aligns closely with \cite{Hoeffding42}'s and \cite{Silvey64}'s conception as described in Section \ref{subsec:puredep}. The rationale for (P\ref{ax:GDPI}) is detailed in Section \ref{subsec:DPI} and is shown to have wide implications. 

\subsection{Pure dependence versus predictability} \label{subsec:puredep}

One can view a vector $(X_1,X_2)$ whose components are independent as a random system with two degrees of freedom, in the sense that $X_1$ and $X_2$ are allowed to vary totally freely. By contrast, any degree of dependence between $X_1$ and $X_2$ necessarily restrains, to some extent, their free variation, reducing {\it de facto} the associated number of degrees of freedom for $(X_1,X_2)$ to strictly smaller than 2. This number can be thought of as the (possibly fractional) number of latent variables able to reproduce in principle the joint behaviour of $(X_1,X_2)$. From that perspective, the opposite of `independence' is thus when $(X_1,X_2)$ has only one degree of freedom, that is, when one single latent variable, say $U$, is able to generate the full covariation of $X_1$ and $X_2$. Formally, this means that there is a function $\Psi:[0,1]\rightarrow \R^2$ such that $(X_1,X_2)=\Psi(U)$. Although we are to see two variables $X_1$ and $X_2$, the underlying probabilistic process is fed by one single source of variability, and $X_1$ and $X_2$ are just the two sides of the same coin. This is essentially what we refer to as {\it pure dependence}. The concept is strongly related to \cite{Hoeffding42}'s `$c$-dependence', and is akin to the joint distribution $F_{12}$ being singular with respect to the product $F_1F_2$ of its margins \citep{Silvey64}, although not exactly equivalent \citep{Durante13}. 
 
\ppn Indeed, the condition $\iint_{\Cs_\Psi} dF_1(x_1)\, dF_2(x_2)=0$ in (P\ref{ax:funct}) implies that $\Cs_\Psi$ must be a `proper curve', in the specific sense that the intersection of $\Cs_\Psi$ with any line $x_1 = t$ or $x_2=t$, $t \in \R$, consists almost surely of at most finitely many points. In particular, it excludes the situations where, although seeded by one single random variable $U$, $(X_1,X_2)=\Psi(U)$ defines a couple of independent random variables. Clearly, this is the case when $\Psi$ defines a constant function of either $X_1$ or $X_2$, making $X_1$ or $X_2$ a degenerate variable hence independent of any other. It is also the case, for instance, when $\Psi$ is a {\it space-filling curve} such as the Peano or the Hilbert curve. Those are known to be surjective and continuous functions from the unit interval $\Is \doteq [0,1]$ onto the unit square $\Is^2$; hence, as $u$ varies from 0 to 1, $\Psi(u)$ visits every single point of $\Is^2$ (see Figure \ref{fig:peano} below). In addition, they have the bi-measure-preserving property \citep[Section 2.7]{Steele97}: for any Borel set $\As \subseteq \Is$, $\lambda_1(\As) = \lambda_2(\Psi(\As))$, where for $q=1,2$, $\lambda_q$ is the Lebesgue measure on $\mathbb{R}^q$. In theory, it is thus possible to generate a bivariate uniform vector $(X_1,X_2)$ on the unit square, hence $X_1 \indep X_2$, by letting $U$ run on $\Is$ and defining $(X_1,X_2) = \Psi(U)$ \citep[p.\,43]{Steele97}. Of course, by definition the image of this function is the whole $\Is^2$, and $\iint_{\Is^2} dF_1(x_1)\, dF_2(x_2)=1$ which violates our definition of pure dependence~(P\ref{ax:funct}).

\ppn Those space-filling curves are special cases of fractal constructions, and the observed independence of $X_1$ and $X_2$ is actually induced by the inherent chaos in the fractal $\Psi$. A fractal is obtained as the limit of a series of iterations reproducing a certain regular pattern at finer and finer resolution. `Shuffles of Min' constructions \citep{Kimeldorf78,Mikusinski92} are of the same nature, while %\cite{Yenigun11} and 
\cite{Zhang18} constructed a similar example based on %`Lissajous curves' and 
a `bisection expanding cross' (Figure~\ref{fig:BET} below); see also \citet[p.\,2287]{Sejdinovic13} who consider sine curves of higher and higher frequencies. Denote by $\Psi_d$ the approximation of the fractal $\Psi$ at resolution $d \in \N$, and $\Cs_d \subset \R^2$ its image. For any finite $d$, $\iint_{\Cs_d} dF_1(x_1)\, dF_2(x_2)=0$, so that defining $(X_1,X_2) = \Psi_d(U)$ would produce a couple of purely dependent variables according to (P\ref{ax:funct}). Now, in the limit $d \to \infty$, their pure dependence suddenly turns into independence, meaning that one can approach arbitrarily closely, in a certain sense, distributions showing independence by distributions showing pure dependence. This ostensible paradox is the core of the discussion in \cite{Zhang18}.

\ppn This, however, is very similar to the following simple case of a degenerate bivariate Gaussian distribution: let $X_1 \sim \Ns(0,1)$ and $X_2 = a X_1$, for $a \neq 0$. Then $|\rho(X_1,X_2)| = 1$ (pure dependence), including for $a$ arbitrarily small. Yet, when $a=0$, $\rho(X_1,X_2)=0$ (independence). %Actually, for any $a \neq 0$, the joint distribution $F_{12}$ of $(X_1,X_2)$ is singular with respect to the distribution under independence ($F_1 \times F_2$), while as soon as $a = 0$, $F_{12}$ immediately turns back into absolutely continuous with respect to $F_1 \times F_2$. 
As $a \to 0$, one would thus approach independence arbitrarily closely by pure dependence. The above described paradox is thus well understood and not much troublesome in a simpler context.

\ppn \cite{Renyi59}'s original Axiom E) requires $D$ to be maximum under `strict dependence', %\footnote{But not `only under', hence the fact that $S$ in (\ref{eqn:maxcor}) may be equal to 1 `too easily'; see discussion in \cite{Bell62}.} 
defined as when there exists a Borel function $\psi_1: \R \to \R$ such that $X_2 = \psi_1(X_1)$, or a Borel function $\psi_2:\R \to \R$ such that that $X_1 = \psi_2(X_2)$; that is Axiom 3 in \citet{Granger04} as well. In other words, one of the variables should be deterministically predictable from the other, that is, what \cite{Lancaster63} defined as `complete dependence'. Yet $X_2$ may be a deterministic function of $X_1$ while giving very little information on $X_1$; e.g., $X_2=\sin(X_1)$ ($\psi_1$ is many-to-one). Clearly asymmetric, this concept seems hardly reconcilable with (P\ref{ax:sym}). In an attempt to symmetrise it, one can request the existence of two functions $\psi_1$ and $\psi_2$ such that $X_2 = \psi_1(X_1)$ {\it and} $X_1 = \psi_2(X_2)$, that is, a one-to-one relationship between $X_1$ and $X_2$. This would reduce any sense of perfect dependence to `mutual complete dependence' \citep{Lancaster63} or even `monotone dependence' \citep{Kimeldorf78}, which appears too restrictive. All in all, if some dependence between $X_1$ and $X_2$ may usually help for predicting $X_2$ from $X_1$ or vice-versa, the concepts of predictability and dependence are indeed distinct.

\ppn Note that the characterisation of `pure dependence' in (P\ref{ax:funct}) is symmetric in $X_1$ and $X_2$, and it admits deterministic predictability as a particular case, in the sense that if $X_2 = \psi_1(X_1)$, one can write 
\[(X_1,X_2) =  (F_1^{-1}(U),\psi_1(F_1^{-1}(U))) \doteq \Psi(U), \quad \text{ for } U \sim \Us_{[0,1]}.  \]
Generally, though, it does not require any of the two variables to be predictable from the other. We note that many popular dependence measures fail to satisfy (P\ref{ax:funct}). In particular, Pearson's correlation and distance correlation take their maximum value only in the case of perfect linear relationship between $X_1$ and $X_2$, while rank-based measures such as Spearman's $\rho$, Kendall's $\tau$ or Hoeffding's $D$ are maximum only in the case of `monotone dependence' (deterministic monotonic relationship between the two variables).

\subsection{Generalised Data Processing Inequality, equitability, margin-freeness and copulas} \label{subsec:DPI}

The {\it Data Processing Inequality} is an important information-theoretic concept \citep[Section 2.8]{Cover06}. It carries the intuitively clear idea that information cannot be gained when a signal goes through a noisy channel. Specifically, if $X_1,X_2$ and $X_3$ are three random variables such that $X_1$ and $X_3$ are independent conditionally on $X_2$ ($X_1 \indep X_3 | X_2$), it establishes that $I(X_1,X_3) \leq I(X_1,X_2)$ where $I(X_i,X_j)$ is the Mutual Information between $X_i$ and $X_j$ ($i,j \in \{1,2,3\}$). Concretely, if $X_2$ is a signal containing some information about $X_1$ and we are only able to see $X_3$, a version of $X_2$ diluted in white noise, then the noisy version $X_3$ is necessarily less informative about $X_1$ than $X_2$. It seems fair to paraphrase this as `there is less dependence between $X_1$ and $X_3$ than between $X_1$ and $X_2$', making it reasonable to ask a dependence measure to satisfy the `generalised Data Processing Inequality' (P\ref{ax:GDPI}). 

\ppn The implications of (P\ref{ax:GDPI}) are actually very deep. In particular, we have the following result:
\begin{proposition} \label{prop:selfequit} A dependence measure $D$ satisfying (P\ref{ax:GDPI}) is such that
\begin{equation} D(X_1,X_2) = D(\psi_1(X_1),\psi_2(X_2)) \label{eqn:selfequit} \end{equation}
for any Borel functions $\psi_1$ and $\psi_2$ such that $X_1 \indep X_2 \ |\ \psi_1(X_1)$ and $X_1 \indep X_2 \ |\ \psi_2(X_2)$.
\end{proposition}
\begin{proof} $X_1 \indep X_2 \ \mid\ \psi_1(X_1) \Rightarrow D(X_1,X_2) \leq D(\psi_1(X_1),X_2)$, by (P\ref{ax:GDPI}). Given $X_1$, $\psi_1(X_1)$ is degenerated, so independent of any other random variable. Thus $\psi_1(X_1) \indep X_2 \ |\ X_1$ and $D(\psi_1(X_1),X_2) \leq D(X_1,X_2)$, again by (P\ref{ax:GDPI}). Hence $D(\psi_1(X_1),X_2) = D(X_1,X_2)$. The second part follows identically, as $X_1 \indep X_2 \ |\ \psi_2(X_2) \Rightarrow \psi_1(X_1) \indep X_2 \ |\ \psi_2(X_2)$. \end{proof}
This property is strongly related to the concept of equitability, which recently came to light for discriminating between dependence measures. In short, a dependence measure is equitable if it returns the same value to equally noisy relationships of different nature. After it was empirically outlined by \cite{Reshef11}, several formal definitions were proposed \citep{Kinney14,Reshef16,Ding17}, and (\ref{eqn:selfequit}) is actually a slight generalisation of \citet{Kinney14}'s definition of {\it self-equitability}. The conditional independence assumptions $X_1 \indep X_2 \ |\ \psi_1(X_1)$ and $X_1 \indep X_2 \ |\ \psi_2(X_2)$ essentially mean that the whole dependence between $X_1$ and $X_2$ can be captured by the functions $\psi_1$ and/or $\psi_2$, and the dependence measure should reflect that. This is the case, for instance, under the regression model $X_2 = \psi_1(X_1) + \varepsilon$, where $\varepsilon \indep X_1$. Then (\ref{eqn:selfequit}) implies that $D(X_1,X_2) = D(\psi_1(X_1),X_2)$, meaning that the value of $D$ is only driven by the signal-to-noise ratio, irrespective of the nature or shape of $\psi_1$. 

\ppn Note that the set of functions $\psi_1,\psi_2$ such that $X_1 \indep X_2 \ |\ \psi_1(X_1)$ and $X_1 \indep X_2 \ |\ \psi_2(X_2)$ is necessarily non-empty. In particular, all strictly monotonic Borel functions $\psi_1$ and $\psi_2$ are such functions, as in that case, conditioning on $\psi_1(X_1)$ or on $\psi_2(X_2)$ is equivalent to conditioning directly on $X_1$ or $X_2$, respectively. Thus, (\ref{eqn:selfequit}) implies that $D$ is invariant to monotonic transformations of $X_1$ and $X_2$,  i.e.,
\begin{equation} D(\psi_1(X_1),\psi_2(X_2)) = D(X_1,X_2) \label{eqn:invmon} \end{equation}
for any two strictly monotonic Borel functions $\psi_1,\psi_2: \R \to \R$. This has often been presented as a fundamental trait of any valid dependence measure: it is Axiom F) in \citet{Renyi59}'s original paper and Axiom 6 in \citet{Granger04}, for instance. A dependence measure satisfying (\ref{eqn:invmon}) is said margin-free, given that the marginal distributions can be arbitrarily distorted by $\psi_1$ and $\psi_2$ without affecting its value. It so appears that (P\ref{ax:GDPI}) is actually a more fundamental property than `margin-freeness', given that (P\ref{ax:GDPI}) $\Rightarrow$ (\ref{eqn:selfequit}) $\Rightarrow$ (\ref{eqn:invmon}).

\ppn Margin-freeness is typically associated to copulas. The copula $C_{12}$ of the continuous vector $(X_1,X_2)$ is the distribution of $(F_1(X_1),F_2(X_2))$ on the unit square $\Is^2$, and is known to capture all the characteristics of $F_{12}$ which are invariant to monotonic transformations of its margins \citep{Schweizer81}. Thus, any dependence measure which is explicitly copula-based, e.g., Spearman's $\rho$, Kendall's $\tau$ or Hoeffding's $D$, is margin-free in the continuous setting. Conversely, for $X_1$ and $X_2$ continuous, any margin-free dependence measure $D(X_1,X_2)$ must admit a representation involving only the copula of $(X_1,X_2)$. Many popular measures of dependence are not copula-based hence they are not margin-free. Besides the obvious example of Pearson's correlation, it is the case of the distance correlation dCor \citep{Szekely07} and the Hilbert-Schmidt independence criterion HSIC \citep{Gretton05}, among others. Problems that this creates were implicitly acknowledged by \citet[Section 4.3]{Szekely09} as they briefly mentioned basing empirical estimation of dCor on the ranks of the observations instead of on the original observations. Likewise, \cite{Poczos12} suggested a copula version of HSIC. Here it is stressed that, beyond margin-freeness, the real issue with measures not copula-based is that they violate (P\ref{ax:GDPI}), whose legitimacy seems difficult to contest.

\section{$\phi$-dependence measures} \label{sec:phidep}

\subsection{Definition and properties} \label{subsec:phidep}

A measure of the dependence in $(X_1,X_2)$ should quantify how much different is the joint distribution $F_{12}$ from the product $F_1 F_2$ of its marginals; see Section \ref{sec:dependence}. Natural candidates for this task are the $\phi$-divergences \citep{Ali66} between $F_{12}$ and $F_1F_2$, viz.
\begin{equation} \Delta_\phi(F_{12}\|F_1F_2) \doteq  \iint_{\R^2} \phi\left(\frac{dF_{12}(x_1,x_2)}{dF_1(x_1)dF_2(x_2)} \right)\,dF_1(x_1)dF_2(x_2) \label{eqn:phidiv} \end{equation}
for some convex function $\phi$ such that $\phi(1) =0$. The $\phi$-divergence family includes most of the common statistical distances between distributions \citep{Liese06}, hence (\ref{eqn:phidiv}) includes many familiar dependence measures. For instance, $\phi(t) = t \log t$ yields the Kullback-Leibler divergence between $F_{12}$ and $F_1 F_2$, which is the Mutual Information $I(X_1,X_2)$.

\ppn Provided that we allow the Radon-Nikodym derivative $dF_{12}/dF_1 dF_2$ to be infinite in case of singularity (see discussion in \citet{Silvey64} and Remark \ref{rem:sing} below), $\Delta_\phi(F_{12}\|F_1F_2)$ is always defined and `existence' (P\ref{ax:exist}) is guaranteed. `Symmetry' (P\ref{ax:sym}) is obvious from the definition. `Weak Gaussian conformity' (P\ref{ax:Gaussconf}) is what \citet{Ali65} established. `Generalised Data Processing Inequality' (P\ref{ax:GDPI}) holds by Theorem 4 of \citet{Kinney14}, which makes any measure (\ref{eqn:phidiv}) automatically margin-free. Indeed, with $X_1$ and $X_2$ both continuous, one has the copula form
\begin{equation} \Delta_\phi(F_{12}\|F_1F_2) = \iint_{\Is^2} \phi\left(\frac{dC_{12}(u_1,u_2)}{du_1du_2} \right)\,du_1\,du_2 = \iint_{\Is^2} \phi\left(c_{12}(u_1,u_2)\right)\,du_1\,du_2, \label{eqn:phidivcop} \end{equation}
through Probability Integral Transform $u_k = F_k(x_k)$, $k=1,2$, where $C_{12}$ is the copula of $(X_1,X_2)$ and $c_{12}$ its density.
\begin{remark} \label{rem:sing}
The copula density $c_{12}$ is obviously defined when $C_{12}$ is absolutely continuous. If $C_{12}$ is singular or has a singular component, then one can think of $c_{12}$ as infinite on the singularity, and defined as the limit of the densities of a sequence of absolutely continuous copulas converging to $C_{12}$ in an appropriate sense \citep[p.\,9]{Ding17}.   
\end{remark}
There remain (P\ref{ax:norm}), (P\ref{ax:indep}) and (P\ref{ax:funct}), compliance to which depends on $\phi$ as follows. 
\begin{proposition} \label{thm:phi}  Let $\phi: (0,+\infty) \to\R$ be convex with $\phi(1)=0$, and call $\varphi_0 = \lim_{t \to 0} \phi(t)$, $\varphi^*_0 = \lim_{t \to 0} t\phi(1/t)$ and $\bar{\varphi} = \varphi_0 + \varphi_0^*$. Then, $\Delta_\phi(F_{12}\|F_1F_2)$ in (\ref{eqn:phidiv}) is such that $ 0 \leq \Delta_\phi(F_{12}\|F_1F_2) \leq \bar{\varphi}$. In addition, if $\phi$ is strictly convex at $t=1$, then $\Delta_\phi(F_{12}\|F_1F_2) = 0$ if and only if $X_1 \indep X_2$ and, provided that $\bar{\varphi} <\infty$, $\Delta_\phi(F_{12}\|F_1F_2) = \bar{\varphi}$ if and only if $X_1$ and $X_2$ are purely dependent in the sense of (P\ref{ax:funct}).
\end{proposition}
\begin{proof} This follows from Theorem 5 in \cite{Liese06}. \end{proof}

%This allows us to sort out which of the $\phi$-dependence measures satisfy the axioms (P\ref{ax:exist})--(P\ref{ax:GDPI}).  are related to whether or not $\bar{\varphi} < \infty$. 
\ppn  `Characterisation of independence' (P\ref{ax:indep}) is thus granted as soon as $\phi$ is strictly convex at $t=1$. If $\bar{\varphi} < \infty$, then `normalisation' (P\ref{ax:norm}) is achieved by obvious linear rescaling $\Delta^\star_\phi(F_{12}\|F_1F_2) \doteq (1/\bar{\varphi}) \Delta_\phi(F_{12}\|F_1F_2)$, and `characterisation of pure dependence' (P\ref{ax:funct}) for $\Delta^\star_\phi(F_{12}\|F_1F_2)$ follows straight from Proposition~\ref{thm:phi}. On the other hand, if $\bar{\varphi} = \infty$, then, even though one can still enforce (P\ref{ax:norm}) through a non-linear transformation, (P\ref{ax:funct}) cannot be made true: there are cases in which $\Delta_\phi(F_{12}\|F_1F_2)$ is infinite while $X_1$ and $X_2$ do not show any sense of strong dependence (see Section \ref{sec:comphi}). In short, a $\phi$-dependence measure fails to comply with (P\ref{ax:funct}) when its baseline version (\ref{eqn:phidiv}) can be infinite -- see comments in \cite{Micheas06} (in particular, their Axiom 9) -- which has %This is easily illustrated through some common choices of function $\phi$ below. 
also serious implications when empirically estimating such a measure \citep[Section 3.1]{Ding17}.

\subsection{Common choices} \label{sec:comphi}

If one takes $\phi(t) = (t-1)^2$, for which choice $\bar{\varphi}=\infty$, $\Delta_\phi(F_{12}\|F_1F_2)$ is Pearson's Mean Square Contingency coefficient $\Phi^2(X_1,X_2)$, allowed to be infinite. It is usually re-normalised as $\Phi^\star(X_1,X_2) \doteq \Phi(X_1,X_2)/(1+\Phi^2(X_1,X_2))^{1/2} \in [0,1]$ so as to agree with $|\rho|$ in the Gaussian case \citep{Renyi59}. However, in general this may equal 1 even when $X_1$ and $X_2$ do not show any sense of strong relationship. Indeed, in the form (\ref{eqn:phidivcop}), $\Phi^2(X_1,X_2)=\iint_{\Is^2} (c_{12}(u_1,u_2)-1)^2\,du_1\,du_2$, and it is known that the copula density $c_{12}$ is not square-integrable as soon as there is the slightest level of tail dependence between $X_1$ and $X_2$ \citep[Theorem 3.3]{Beare10}. This observation was already used by \citet[p.\,150]{Hoeffding42} for calling into question the usefulness of $\Phi^2$ as a measure of dependence.

\ppn Likewise, for $\phi(t) = t \log t$, one gets $\bar{\varphi} = \infty$, and indeed the Mutual Information $I(X_1,X_2)$ can be infinite. It can be re-normalised as $I^\star(X_1,X_2) \doteq \left[1-\exp(-2I(X_1,X_2))\right]^{1/2} \in [0,1]$ so as to agree again with $|\rho|$ in the Gaussian case \citep{Linfoot57}. Yet, it cannot unequivocally characterise pure dependence. For instance, for some $a \in (0,1)$ and $m \in \N$, consider the copula density 
\[c(u,v) = \left\{\begin{array}{l l} \frac{1}{1-a} & \text{ if } (u,v) \in [0,1-a)^2  \\ 
\frac{m}{a} & \text{ if } (u,v) \in [1-a+(\nu-1)a/m,1-a+\nu a/m)^2,\ \nu=1,2,\ldots,m, \\
0 & \text{ elsewhere. } \end{array} \right. \]
Then it can be checked that 
\[I(X_1,X_2)=\iint_{\Is^2} c(u_1,u_2) \log c(u_1,u_2)\,du_1 \,du_2 = -(1-a) \log(1-a) - a \log a + a \log m. \]
As $a \to 0$, the copula density is constant over (almost) the whole of $\Is^2$, suggesting a situation of near-independence. However, $I(X_1,X_2)$ can be made arbitrarily large by growing $m$ such that $a  \log m \to \infty$. 

\ppn By contrast, the choice $\phi(t) = |t-1|$ yields $\bar{\varphi} =2 < \infty$. The associated (rescaled) dependence measure $\Delta^\star_\phi(F_{12}\|F_1F_2) = (1/2) \Delta_\phi(F_{12}\|F_1F_2)$ consequently satisfies (P\ref{ax:exist})--(P\ref{ax:GDPI}). It is actually \cite{Silvey64}'s $\Delta$, renamed `robust copula dependence' measure in \cite{Ding17} when in copula form (\ref{eqn:phidivcop}). However, no root-$n$ consistent estimator of that measure seems available. In the next section, we explore a particular measure which lies at the intersection  between two common families of $\phi$-divergences (see Appendix \ref{App:phidiv}), and we provide a simple root-$n$ consistent estimator.

\section{The Hellinger correlation}

If one takes $\phi(t) = (t^{1/2}-1)^2$, for which $\bar{\varphi} = 2 < \infty$, the corresponding (rescaled) measure 
\begin{align}
 \Delta^\star_\phi(F_{12}\|F_1F_2) & = \frac{1}{2} \iint_{\R^2}\left(\sqrt{\frac{dF_{12}(x_1,x_2)}{dF_1(x_1)dF_2(x_2)}} -1\right)^{2}\,dF_1(x_1)dF_2(x_2) \notag \\
 = \frac{1}{2} \iint_{\R^2} & \left(\sqrt{dF_{12}(x_1,x_2)} -\sqrt{dF_1(x_1)dF_2(x_2)}\right)^{2} \doteq \Hs^2(X_1,X_2)\label{eqn:Hellfull}%\\
%& = 1 - \iint_{\R^2} \sqrt{dF_{12}(x_1,x_2)} \sqrt{dF_1(x_1)dF_2(x_2)}. \label{eqn:Batt}
\end{align}
satisfies (P\ref{ax:exist})--(P\ref{ax:GDPI}). We denote this measure $\Hs^2(X_1,X_2)$, or simply $\Hs^2$, as it is the squared Hellinger distance between $F_{12}$ and $F_1F_2$. %The Hellinger distance \citep{Bhattacharyya43} enjoys many desirable properties as a statistical distance between probability distributions \citep{Gibbs02}. %E.g.\ \cite{Beran77} developed a competitor to the parametric maximum likelihood estimator as the minimiser of the Hellinger distance within the parametric model, and proved its properties of efficiency and robustness. 
Under the copula form (\ref{eqn:phidivcop}), it is
\begin{equation} \Hs^2   = \frac{1}{2}\iint_{\Is^2} \left(\sqrt{c_{12}}(u_1,u_2)-1\right)^2\,du_1\,du_2  = 1- \iint_{\Is^2} \sqrt{c_{12}}(u_1,u_2)\,du_1\,du_2 \doteq 1-\Bs, \label{eqn:Helldist2contcop} \end{equation}
where $\Bs$ is the {\it Bhattacharyya affinity coefficient} \citep{Bhattacharyya43} between the copula density $c_{12}$ and the independence copula density identically equal to $1$ on $\Is^2$.

\ppn In the bivariate Gaussian case, it can be checked that $\Hs^2 = 1-(2(1-\rho^2)^{1/4})/(4-\rho^2)^{1/2} \doteq  h(\rho)$, a strictly increasing function of the correlation $\rho$ -- in agreement with `weak Gaussian conformity' (P\ref{ax:Gaussconf}). This suggests to consider the measure on the transformed scale $\eta = h^{-1}(\Hs^2)$. As $h^{-1}$ is a bijection from $[0,1]$ to $[0,1]$, it preserves all (P\ref{ax:exist})--(P\ref{ax:GDPI}) for $\eta$. Direct algebra yields
%\begin{equation} \eta = \frac{2}{\Bs^2}\sqrt{\Bs^4 +\sqrt{4-3\Bs^4}-2} \label{eqn:Hellcorr}\end{equation}
\begin{equation} \eta = (2/\Bs^2)\left\{\Bs^4 +(4-3\Bs^4)^{1/2}-2\right\}^{1/2}, \label{eqn:Hellcorr}\end{equation}
our proposed measure of dependence. We call $\eta \doteq \eta(X_1,X_2)$ the {\it Hellinger correlation} between $X_1$ and $X_2$, given that it is defined so that $\eta = |\rho|$ in the bivariate Gaussian case (`strong Gaussian conformity'). This greatly facilitates interpretation as the value of $\eta$ can easily be appreciated on the familiar Pearson's correlation scale: a Hellinger correlation $\eta$ equal to $\eta_0\in[0,1]$ represents a dependence of the same strength as in a bivariate Gaussian vector whose Pearson's correlation is $\rho=\eta_0$.

\section{Empirical estimation and significance} \label{sec:est} % and independence testing} 

\subsection{Background}

Measuring dependence by $\Hs^2$ has been considered before, e.g., by \cite{Granger04} %Although, it has remained rather inconspicuous, a reason for that might be that ``{\it (...) for the Hellinger distance, it seems awkward in practice to estimate the density function}'' \citep[p.\,365]{Tsukahara05}. 
who proposed an estimator based on kernel density estimation and numerical integration -- see also \cite{Su08}. However, the obtained estimate heavily depends on the bandwidths used in the kernel estimators \citep{Skaug96}, an appropriate choice of which in practice remaining problematic. \citet{Rosenblatt75}, \citet{Hong05} and \citet{Ding17} face the same issue when estimating their proposed measure; see also comments in \cite{Zeng18}. Of course, difficulty in estimating a measure from empirical data seriously limits its practical reach.

%Following the recommendation of \cite{Vapnik98}, though, `{\it one should avoid solving more difficult intermediate problems when solving a target problem}'. 
\ppn  Now, if knowledge of $dF_{12}$, $dF_1$ and $dF_2$ in (\ref{eqn:Hellfull}), or of the copula density $c_{12}$ in (\ref{eqn:Helldist2contcop}), implies knowledge of $\Hs^2$, the contrary is not true: one cannot recover $dF_{12}$, $dF_1$ and $dF_2$ or $c_{12}$ from $\Hs^2$ alone. Empirical estimation of $\Hs^2$ (or any function thereof, such as $\Bs$ or $\eta$) should thus not be based on the more difficult task of estimating the individual densities. Indeed it has been known, at least since \cite{Kozachenko87}, that one can consistently estimate some $\phi$-divergences without consistently estimating the underlying distributions; \cite{Berrett18} offer a recent review of such ideas for entropy estimation. %consistent direct estimation of $\phi$-divergences could be achieved making use of $k$-nearest-neighbour ideas; see \cite{Singh16} and references therein. 
Below we suggest a simple estimator of $\Bs = \iint \sqrt{c}_{12}$ along that line, subsequently producing an estimator of $\eta$ through (\ref{eqn:Hellcorr}).

\subsection{Basic estimator and asymptotic properties}

Let $\{\XX_1,\ldots,\XX_n\}$, where $\XX_i = (X_{i1},X_{i2}) \in \R^2$, be a random sample from $F_{12}$. For $i \in \{1,\ldots,n\}$, call  $\UU_i = (U_{i1},U_{i2}) = (F_1(X_{i1}),F_2(X_{i2}))$, the corresponding observations from the copula $C_{12}$, and $\widehat{\UU}_i = (\hat{U}_{i1},\hat{U}_{i2}) = (\hat{F}_{n,1}(X_{i1}),\hat{F}_{n,2}(X_{i2}))$ the `pseudo'-observations obtained from $\hat{F}_{n,k}(x) \doteq (1/(n+1)) \sum_{i=1}^n \indic{X_{ik}\leq x}$, $k=1,2$, as is customary in the copula literature. Let $R_i = \min_{j \neq i} \|\UU_j-\UU_i\|_2$, the Euclidean distance between $\UU_i$ and its closest neighbour, and its `pseudo'-version $\hat{R}_i = \min_{j \neq i} \|\widehat{\UU}_j-\widehat{\UU}_i\|_2$. 

\ppn Our first simple and entirely data-based estimator for $\Bs$, not involving any user-defined parameter, is
\begin{equation} \widehat{\Bs} = \frac{2\sqrt{n-1}}{n}\sum_{i=1}^n \hat{R}_i \label{eqn:B12hat}, \end{equation}
which is the feasible version of the oracle estimator $\widetilde{\Bs} \doteq (2\sqrt{n-1}/n) \sum_{i=1}^n R_i$. See \citet[Remark 3.1]{Leonenko08} for the motivation behind this estimator.   %It is known that the distances between nearest neighbours among $n$ observations in $\R^d$ are $O_P(n^{-1/d})$ \citep[Proposition 2]{Ranneby05}. 
In terms of measuring dependence, the intuition is the following. In case of pure dependence ($\Bs = 0$), the $\UU_i$'s fall exactly on a curve. The $R_i$'s are then essentially univariate spacings, known to be of order $O_P(n^{-1})$ \citep{Pyke65}. This yields $\widetilde{\Bs} \to 0$ in probability. Now gradually relaxing dependence amounts to allowing some play around that curve for the $\UU_i$'s. As these get more and more room to move apart, the $R_i$'s globally increase, and so does $\widetilde{\Bs}$. Ultimately, when the $\UU_i$'s get totally free (independence, $\Bs = 1$), they uniformly cover $\Is^2$ and maximally occupy their allowed space. The $R_i$'s are globally as large as can be, and so is $\widetilde{\Bs}$. The latitude given to the $\UU_i$'s for covering $\Is^2$ reflects the number of degrees of freedom in $(X_1,X_2)$; see Section~\ref{subsec:puredep}.

\ppn The $L_2$-consistency of $\widetilde{\Bs}$ follows straight from \citet[Section 3.1]{Leonenko08}. \citet[Lemma 1]{Aya18} and \citet[Theorems 1 and 2]{Ebner18} further establish the root-$n$ consistency and the asymptotic normality of $\widetilde{\Bs}$, if the copula density $c_{12}$ is bounded on $\Is^2$. In Appendix \ref{App:margtrans} we show how to relax this restrictive assumption by combining marginal transformations \citep{Geenens17} and results of \cite{Singh16}.  \citet[Lemmas 2.1 and 2.2]{Deheuvels09} help to establish that $\widehat{\Bs}-\widetilde{\Bs} = O_P(n^{-1/2})$, meaning that the above root-$n$ consistency of the oracle estimator $\widetilde{\Bs}$ carries over to the feasible $\widehat{\Bs}$.

\subsection{Normalisation} \label{subsec:norm}

In finite samples, however, $\widehat{\Bs}$ suffers from two serious defects (as would $\widetilde{\Bs}$). First, it is heavily biased due to the boundedness of the support $\Is^2$ of $C_{12}$ \citep{Liitiainen10}. Second, although $\Bs = \iint \sqrt{c}_{12} \leq \iint (\sqrt{c}_{12})^2 = 1$ by Cauchy-Schwarz, it may happen that $\widehat{\Bs} > 1$, implying a meaningless negative estimate of $\Hs^2$ and precluding ulterior estimation of $\eta$. 

\ppn Now, evidently $\sqrt{c}_{12} \in L_2(\Is^2)$. Hence from any orthonormal basis of $L_2(\Is)$, say $\{b_0(u),b_1(u),\ldots\}$ with $b_0(u) \equiv 1$, one can form a tensorised orthonormal basis for $L_2(\Is^2)$ and write the expansion $\sqrt{c}_{12}(u_1,u_2) = \sum_{k=0}^\infty \sum_{\ell=0}^\infty \beta_{k\ell} b_k(u_1)b_\ell(u_2)$, where $\beta_{k\ell} = \iint_{\Is^2} \sqrt{c}_{12}(u_1,u_2) b_k(u_1)b_\ell(u_2)\,du_1 du_2$. In particular, see that $\Bs = \beta_{00}$. Similarly to (\ref{eqn:B12hat}), 
%\[\hat{\beta}_{k\ell} = \frac{2\sqrt{n-1}}{n}\sum_{i=1}^n S_i \sqrt{\xi_1}(T_{i1}) \sqrt{\xi_2}(T_{i2})b_k(U_{i1})b_\ell(U_{i2}). \]
\begin{equation} \hat{\beta}_{k\ell} = \frac{2\sqrt{n-1}}{n}\sum_{i=1}^n \hat{R}_i b_k(\hat{U}_{i1})b_\ell(\hat{U}_{i2}), \qquad k,\ell\in\mathbb{N},\label{eqn:tildebetakl} \end{equation}
is a root-$n$ consistent estimator of $\beta_{k\ell}$ \citep[Lemma 1]{Aya18}. For appropriate cut-off values $K$ and $L$, define $\widehat{\sqrt{c}}_{12}(u_1,u_2)= \sum_{k=0}^K \sum_{\ell=0}^L \hat{\beta}_{k\ell} b_k(u_1)b_\ell(u_2)$ an orthogonal series estimator for $\sqrt{c}_{12}$, and see that $\iint_{\Is^2} \left(\widehat{\sqrt{c}}_{12}\right)^2(u_1,u_2) \,du_1 du_2 = \sum_{k=0}^K \sum_{\ell=0}^L \hat{\beta}^2_{k\ell}$. This suggests to re-normalise $\widehat{\Bs} = \hat{\beta}_{00}$ as
\begin{equation} \widehat{\Bs}_{\text{\scriptsize $K$\!$L$}} =  \widehat{\Bs}\Big/\left(\sum_{k=0}^K \sum_{\ell=0}^L \hat{\beta}^2_{k\ell}\right)^{1/2}. \label{eqn:Btildestar}\end{equation}
Not only this shrinkage guarantees $\widehat{\Bs}_{\text{\scriptsize $K$\!$L$}} \in [0,1]$ always, it reduces the variance and mostly takes care of the boundary effect as well. Intuitively, when $\UU_i$ is close to the boundary of $\Is^2$, its neighbourhood is empty of data on one side and $\hat{R}_i$ tends to be larger than what it should be. This makes (\ref{eqn:B12hat}) overestimate $\Bs$ and induces bias. But then (\ref{eqn:tildebetakl}) overestimates each $\beta_{k\ell}$ to the same extent, and the induced biases mostly cancel each other out in the ratio (\ref{eqn:Btildestar}). Finally, (\ref{eqn:Btildestar}) is plugged into (\ref{eqn:Hellcorr}) to define the estimator
\begin{equation} \widehat{\eta}_{\text{\scriptsize $K$\!$L$}}^{} = (2/\widehat{\Bs}^{2}_{\text{\scriptsize $K$\!$L$}})\left\{\widehat{\Bs}^{4}_{\text{\scriptsize $K$\!$L$}} +\left(4-3\widehat{\Bs}^{4}_{\text{\scriptsize $K$\!$L$}}\right)^{1/2}-2\right\}^{1/2}, \quad K,L\in\N. \label{eqn:empHellcorr} \end{equation}

\ppn The following simulation showcases the performance of this estimator. Random samples of size $n=500$ were generated from a bivariate Gaussian distribution with correlation (a) $\rho=0.4$ and (b) $\rho=0.8$. On each of them, the basic estimator $\hat{\eta}_0$, obtained by plugging $\widehat{\Bs}$ in \eqref{eqn:Hellcorr}, was computed, together with normalised versions $\hat{\eta}_{11}$, $\hat{\eta}_{22}$ and $\hat{\eta}_{33}$. The basis $\{b_0(u),b_1(u),\ldots\}$ was formed by the normed Legendre polynomials shifted to $[0,1]$. The returned estimates are shown by the boxplots in Figure \ref{fig:BP}. In the case $\rho=0.4$, though, around 56\% of the initially returned estimates $\widehat{\Bs}$ were found greater than 1, precluding estimation of $\eta$. So the boxplot at the extreme left only represents the 44\% of the estimates which could be computed.%, where the influence of the cut-offs $K$ and $L$ can be appreciated. 

\begin{figure}[H]
\centering
\includegraphics[width=0.8\textwidth]{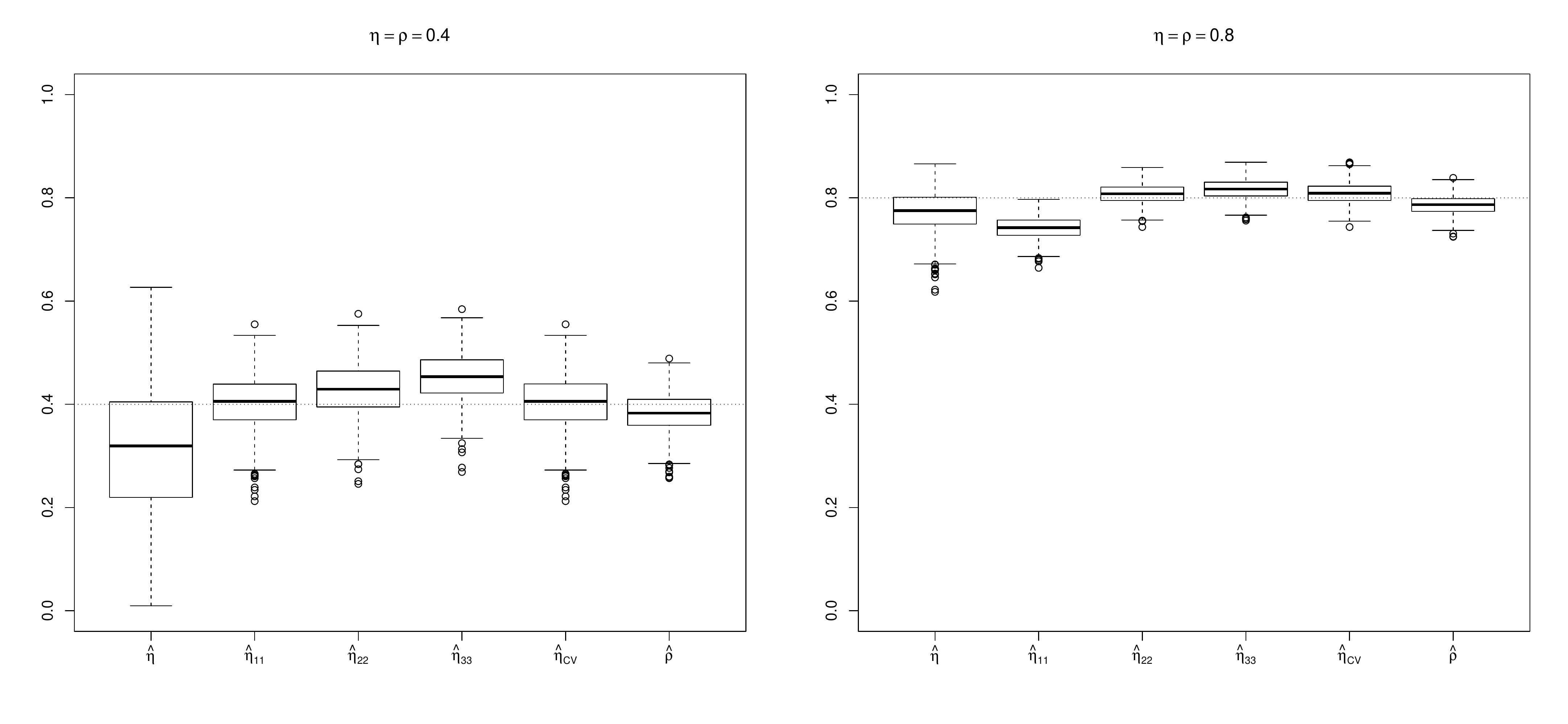}
\caption{1000 replications of 6 estimators of $\eta$ on bivariate Gaussian samples of size $n=500$ with $\rho = 0.4$ (left) and $\rho=0.8$ (right).}
\label{fig:BP}
\end{figure}

\ppn The reduction in both bias and variance brought by the normalisation is obvious. For $\rho=0.4$, one extra term in each direction ($K=L=1$) in the expansion for $\sqrt{c}_{12}$ is enough, as $\sqrt{c}_{12}$ is rather flat and well approximated by a low-degree polynomial. For $\rho=0.8$, slightly more terms should be included as $\sqrt{c}_{12}$ tends to peak in the corners of $\Is^2$. In order to keep the estimator totally data-driven, we suggest in Appendix \ref{App:CV} a novel and explicit cross-validation criterion easy to minimise. The estimator computed with the values of $K$ and $L$ minimising that criterion is marked $\hat{\eta}_\text{CV}$ in Figure \ref{fig:BP}. The proposed cross-validation criterion consistently identifies the suitable level of approximation $K$ and $L$ and produces a reliable estimator of $\eta$. Estimators (\ref{eqn:tildebetakl})-(\ref{eqn:Btildestar})-(\ref{eqn:empHellcorr}), as well as the suggested cross-validation procedure, have been efficiently implemented in the freely available R package \texttt{HellCor}.

\ppn For comparison, the empirical Pearson's correlation $\hat{\rho}$ was computed on the same samples, being here some sort of `gold standard' given that in the considered bivariate Gaussian vectors, $\eta = \rho$. It is seen that $\hat{\eta}_\text{CV}^{}$ is less biased than $\hat{\rho}$, and even does better in terms of Mean Squared Error for $\rho=0.8$ (Table \ref{tab:MSE}). Strikingly, the proposed estimator of $\eta$ is on par with the classical estimator specifically designed for capturing the dependence of linear nature peculiar to bivariate Gaussian vectors.

\begin{table}[H]
\centering
\begin{tabular}{|r| c c | c c |}
\hline
 & bias$(\hat{\eta}_\text{CV})$ & MSE$(\hat{\eta}_\text{CV})$ & bias$(\hat{\rho})$ & MSE$(\hat{\rho})$  \\
\hline
$\rho = 0.4$ & 0.003 & 0.0023 & -0.017 & 0.0017 \\
$\rho= 0.8$ & 0.009 & 0.00051 & -0.013 & 0.00053 \\
\hline
\end{tabular}
\caption{Bias and Mean Squared Error of $\hat{\eta}_\text{CV}$ and $\hat{\rho}$ (empirical correlation) in bivariate Gaussian vectors with correlations $\rho=0.4$ and $\rho=0.8$ ($n=500$).}
\label{tab:MSE}
\end{table}

\ppn Naturally, the Hellinger correlation would capture dependence of any other nature as well. This is illustrated in Figure \ref{fig:HHGplots}, showing 15 random samples of size $n=500$ generated from the 15 scenarii considered in \citet[Figure 4]{Heller16}. The estimator $\hat{\eta}_{\text{CV}}$ with expansion in the Legendre basis and cross-validation cutt-offs, from now on denoted simply $\hat{\eta}$, was computed on each of them. The estimated values of $\eta$ are shown on top of each plot, attesting an interesting descriptive ability. In order to get an idea of the accuracy of estimation, we have replicated $M=1,000$ times each of these 15 scenarii (with $n=500$) and estimated $\eta$ on each of them. Table \ref{tab:etahat} shows the main characteristics of the sampling distribution of $\hat{\eta}$ for each case. The estimation is remarkably accurate, at the exception of when the dependence is harder to detect like in scenarii [14] and [15].

\begin{figure}[H]
	\centering
	\includegraphics[width=0.99\textwidth]{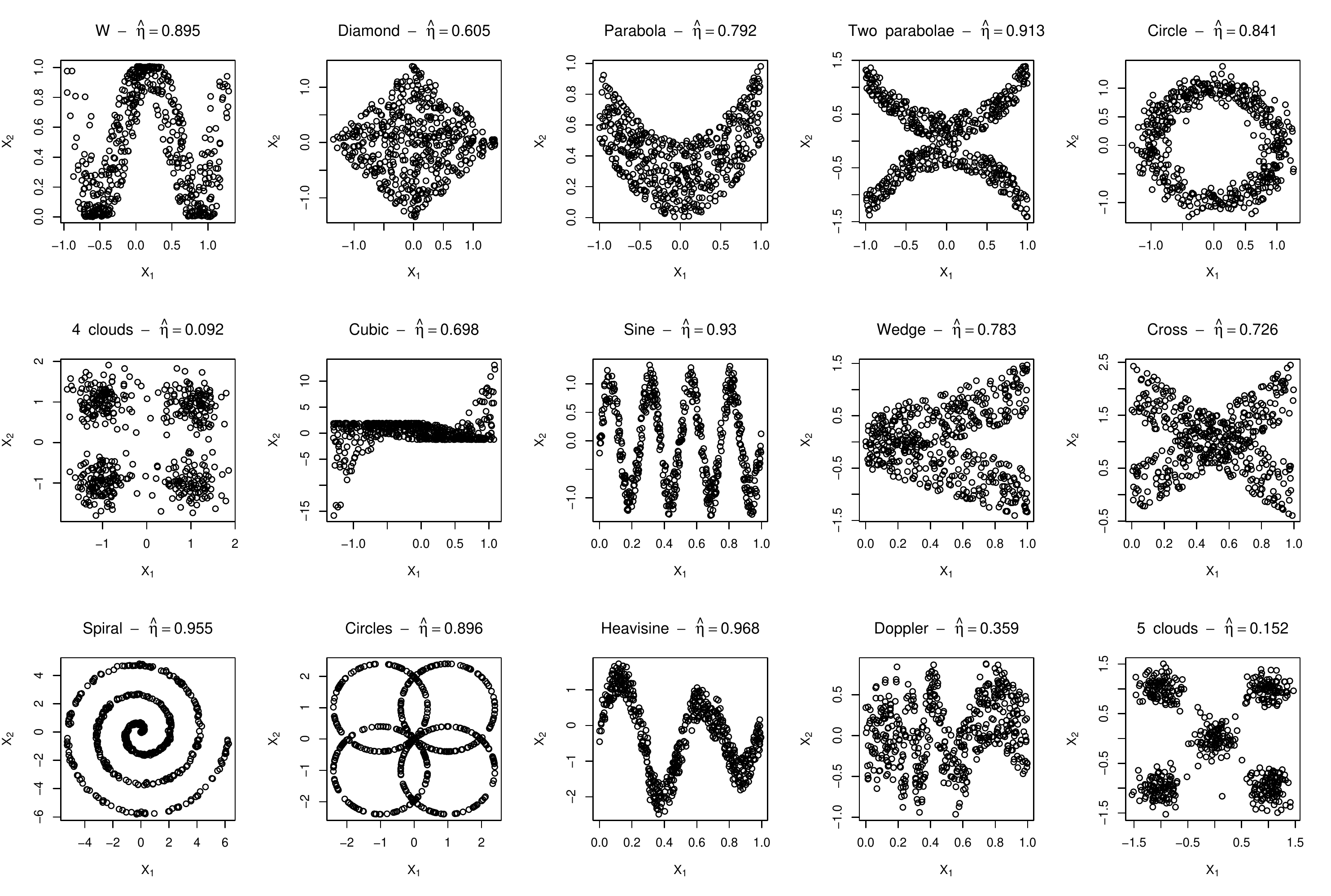}
	\caption{Fifteen random samples of size $n=500$ generated from the 15 scenarii of \citet[Figure 4]{Heller16}. The empirical Hellinger correlations are shown on top.}
	\label{fig:HHGplots}
\end{figure}

\begin{table}[H]
 \centering
\begin{tabular}{|r l | c c c c c|}
\hline
 & & \text{mean} & \text{st.dev.} & \text{min} & \text{median} & \text{max} \\
\hline
\text{[1]} &\text{ W} &0.894&0.007&0.852&0.894&0.924 \\
\text{[2]} &\text{ Diamond} &0.599&0.022&0.526&0.599&0.659 \\
\text{[3]} &\text{ Parabola}&0.798&0.018&0.742&0.802&0.839\\
\text{[4]} &\text{ Two Parabolae}&0.912&0.008&0.890&0.911&0.936\\
\text{[5]} &\text{ Circle}&0.839&0.019&0.778&0.844&0.881\\
\text{[6]} &\text{ 4 clouds}&0.080&0.034&0.010&0.076&0.219\\
\text{[7]} &\text{ Cubic}&0.746&0.028&0.596&0.747&0.825\\
\text{[8]} &\text{ Sine}&0.920&0.011&0.897&0.923&0.940\\
\text{[9]} &\text{ Wedge}&0.755&0.024&0.673&0.754&0.812\\
\text{[10]} &\text{ Cross}&0.734&0.018&0.682&0.736&0.785\\
\text{[11]} &\text{ Spiral}&0.957 & 0.005 & 0.927 & 0.957 & 0.970\\
\text{[12]} &\text{ Circles}&0.914&0.012&0.863&0.918&0.937\\
\text{[13]} &\text{ Heavysine}&0.964&0.002&0.955&0.964&0.972\\
\text{[14]} &\text{ Doppler}&0.461&0.146&0.129&0.428&0.731\\
\text{[15]} &\text{ 5 clouds}&0.136&0.159&0.010&0.092&0.785\\
\hline
\end{tabular}
\caption{Mean, standard deviation, minimum, median and maximum of $M=1,000$ estimates of $\eta$ for the 15 scenarii shown in Figure \ref{fig:HHGplots} ($n=500$).}
\label{tab:etahat}
\end{table}

Finally, we generated random samples of size $n=500$ and $n=5,000$ on the Peano curve (Figure \ref{fig:peano}) and on \cite{Zhang18}'s `bissection expanding cross' (Figure \ref{fig:BET}), at increasing resolution $d$ (see Section \ref{subsec:puredep}). For any finite $d$, the joint distribution of $(X_1,X_2)$ is concentrated on a curve (in the sense of (P\ref{ax:funct})) and hence exhibits perfect dependence (population Hellinger correlation $\eta=1$). However, that dependence is harder and harder to detect empirically as the resolution increases: for the same sample size, the length of the `curve' increases, hence the points are more distant to one another and it becomes difficult to capture their perfect alignment. This explains why the empirical Hellinger correlation tends to fade as $d$ increases (for $n$ fixed). Increasing the sample size pushes the empirical Hellinger correlation towards $\eta =1$ for any finite $d$, and towards $\eta = 0$ for $d=\infty$ -- as explained in Section \ref{subsec:puredep}, $X_1$ and $X_2$ are indeed independent at the limit $d=\infty$. For comparison, we have also computed the empirical distance correlation ($\widehat{\dCor}$) for each of those samples. The observed values of $\widehat{\dCor}$ remain very low across all scenarii, actually close to its value for the independent sample ($d=\infty$), even for $n=5,000$. The observed value is significantly different to 0 (at level $\alpha = 0.05$) only for resolution $d=1$, and $d=2$ for $n=5,000$ in the `bissection expanding cross' case (obtained from $200$ permutations, from \texttt{dcor.test} in the R package \texttt{energy}).  We have also computed a variety of other dependence measures or dependence test statistics (not shown), such as Hoeffding's $D$ \citep{Hoeffding48}, HSIC \citep{Pfister18}, or the HHG test statistics \citep{Heller16}, with essentially the same results: none is able to detect dependence beyond resolution $d=1$.

\begin{figure}[H]
\centering
\includegraphics[width=\textwidth]{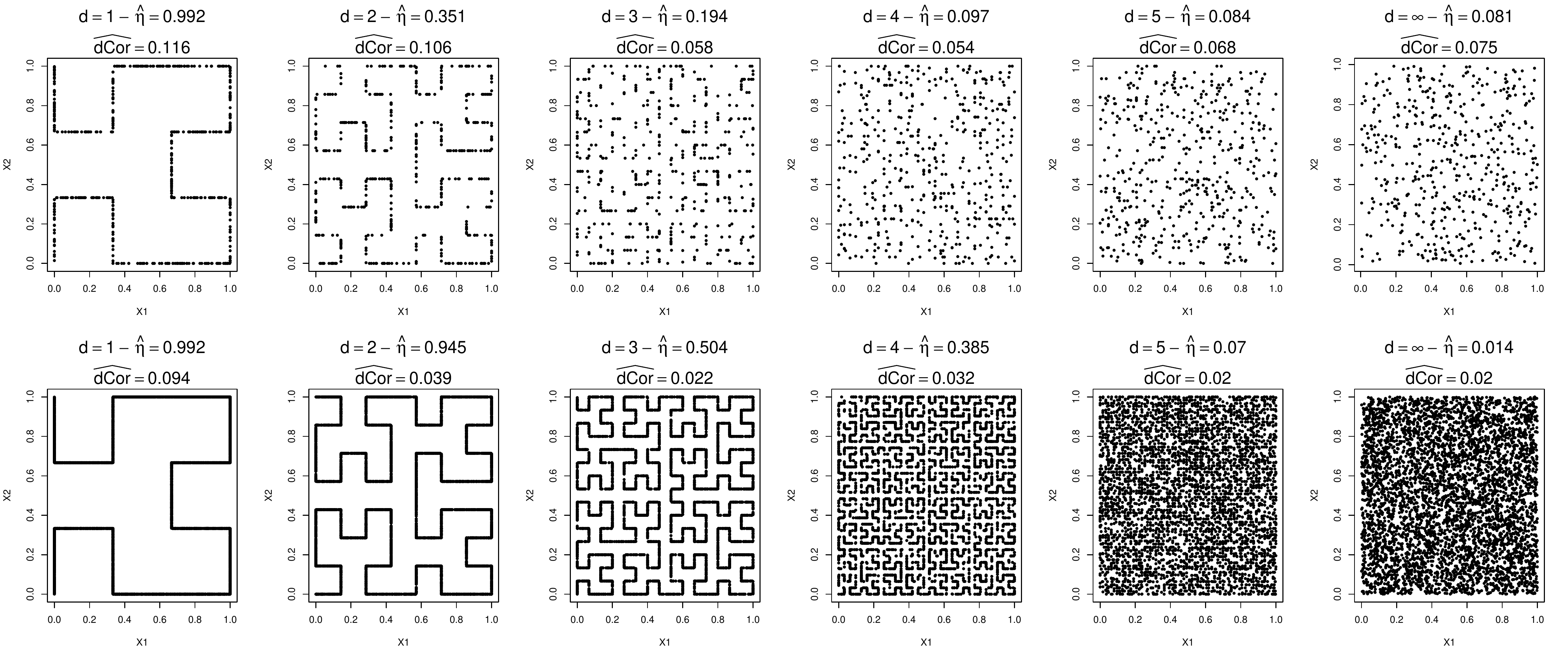}
\caption{Typical random samples of $n=500$ (top row) and $n=5000$ (bottom row) points uniformly distributed on the Peano curve at resolution $d=1,2,3,4,5$ and $d=\infty$ (independence). The empirical Hellinger correlations are shown on top, as well as the empirical distance correlations ($\widehat{\dCor}$).}
\label{fig:peano}
\end{figure}

\begin{figure}[H]
\centering
\includegraphics[width=\textwidth]{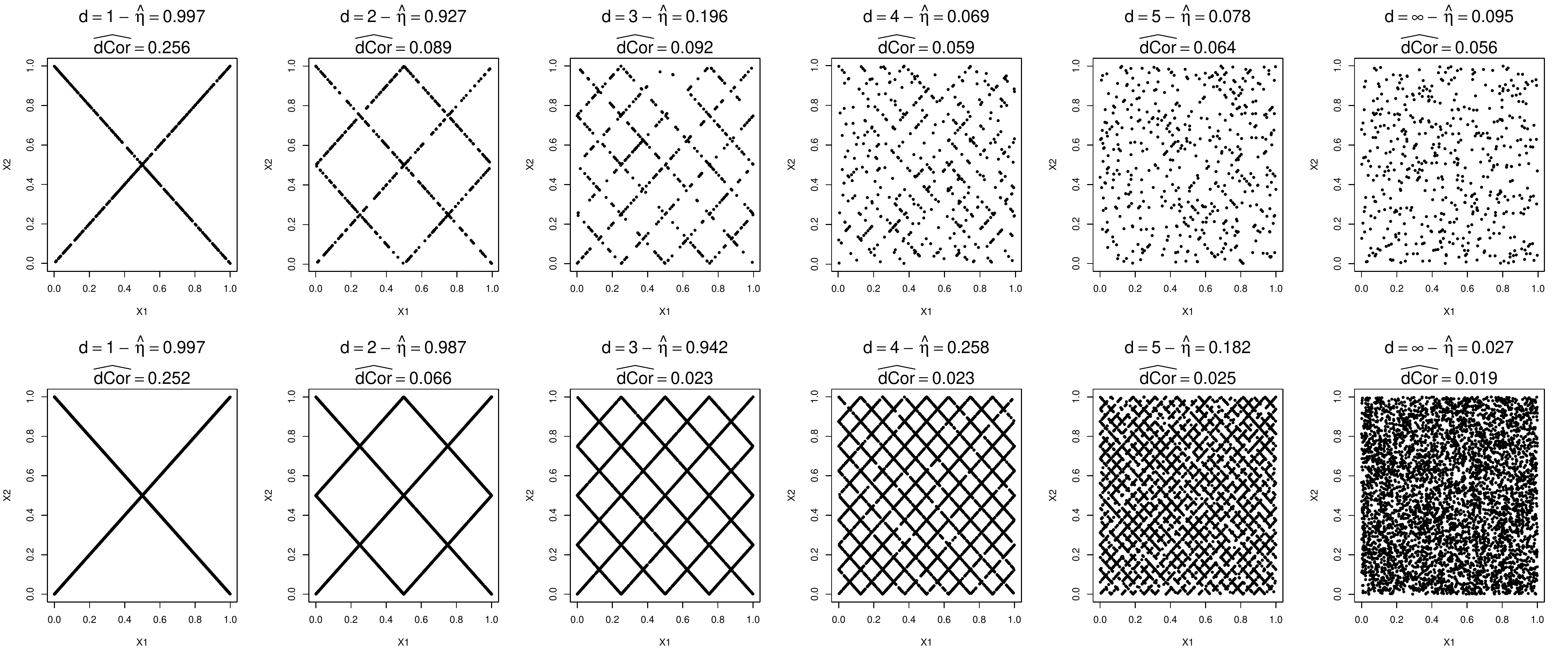}
\caption{Typical random samples of $n=500$ (top row) and $n=5000$ (bottom row) points uniformly distributed on the `bissection expanding cross' \citep{Zhang18} at resolution $d=1,2,3,4$ and $d=\infty$ (independence). The empirical Hellinger correlations are shown on top, as well as the empirical distance correlations ($\widehat{\dCor}$).}
\label{fig:BET}
\end{figure}

\subsection{Significance} \label{sec:indeptest}

The statistical significance of the empirical Hellinger correlation computed on a given data set can easily be obtained for any sample size by Monte-Carlo simulations. Indeed, if $X_1$ and $X_2$ are independent, then $C_{12}$ is bivariate uniform on $\Is^2$. One can then simulate a large number of bivariate uniform samples of any size $n$ and compute $\hat{\eta}$ on each of those, which would allow arbitrary close approximation to the exact sampling distribution of $\hat{\eta}$ for that $n$ in the case of independent variables. For instance, from $M=10,000$ independent bivariate uniform samples of size $n=500$ and $n=5,000$, we have obtained, respectively, critical values $\bar{\eta}_\text{c} = 0.146$ and $\bar{\eta}_\text{c} = 0.047$ (exact up to Monte Carlo error) at significance level $\alpha = 0.05$. All the observed empirical Hellinger correlations in Figure \ref{fig:HHGplots} are thus statistically significant, except under the `4 clouds' scenario [6] in which case $X_1$ and $X_2$ are indeed independent \citep{Heller16}. For the Peano and `bissection expanding cross' scenarii (Figures \ref{fig:peano} and \ref{fig:BET}), dependence is detected up to resolution $d=3$ for $n=500$. At that sample size, it becomes hardly possible to visually make any distinction between the samples for $d=4,5$ and $d=\infty$ (independence), indeed. For $n=5,000$, on the other hand, dependence is detected up to resolution $d=5$.

\ppn Naturally, checking the significance (i.e., non-nullity) of a measure of dependence is in many ways akin to testing for independence. Now, bearing in mind the discussion in Section \ref{sec:intro}, it seems clear that $\hat{\eta}$ might not be a test statistic leading to the most powerful test. In particular, the observations made in Sections \ref{subsec:phidep}-\ref{sec:comphi} indicate that there exist statistics taking infinite values even in the absence of strong dependence between $X_1$ and $X_2$. There is no doubt that those would produce much more powerful test procedures than $\hat{\eta}$, and the problem of testing is thus not investigated in further detail here. Yet, the previous results indicate that, in some situations, the Hellinger correlation may still show some interesting ability to highlight dependence compared to other, more classical choices of dependence measure. This is confirmed by the following application.

\section{Application: Coral fish, seabirds and reef productivity} \label{sec:app}

As an application we consider data from a recent study on coral reef productivity described in \cite{Graham18}. The density of some species of seabirds and coral reef fish around $n=12$ islands of the Chagos Archipelago (British Indian Ocean Territory) was recorded; see Table~\ref{tab:rats}. Those seabirds forage and feed in the open ocean, far from reefs, and their number around a given island largely depends on the presence or absence of rats on it. So they belong to a different ecosystem to the fish. Indeed, empirical Pearson's correlation between the densities of seabirds and fish around the 12 islands under study is $\hat{\rho} = 0.374$, not significantly different to 0 ($p=0.231$). Likewise, the empirical distance correlation \citep{Szekely07} is $\widehat{\dCor} = 0.526$, not significant ($p=0.179$ based on $200$ permutations, from \texttt{dcor.test} in the R package \texttt{energy}). 

\begin{table}[h]
	\centering
	\begin{tabular}{l|rrrrrrrrrrrr}
		{Island} & 1 & 2 & 3 & 4 & 5 & 6 & 7 & 8 & 9 & 10 & 11 & 12 \\
		\hline
		{Seabirds} & 1 & 3702 & 183 & 973 & 1161 &   2 & 1427 &   0 &   3 &  15 &   4 &  1  \\
		{Fish} &  194 & 278 & 279 & 300 & 281 & 244 & 300 & 245 & 212 & 275 & 301 & 265 \\
		%\hdashline
	      %{Nitrogen} &  0.2 & 424.8 & 119.3 & 172.4 & 213.5 &   1.2 & 206.4 &   0.0 &   0.53 &   6.1 &   2.4 &   0.2 \\
		%{Nitrogen input} &  0.22 & 424.83 & 119.33 & 172.44 & 213.46 &   1.16 & 206.42 &   0.04 &   0.53 &   6.06 &   2.44 &   0.16 \\
		\hline
	\end{tabular} \caption{Density of seabirds and fish (individuals per hectare of island, truncated to the nearest integer) around $n=12$ islands of the Chagos Archipelago.
	}  
\label{tab:rats}
\end{table}

\ppn This failure to evidence any significant dependence goes, however, against the report of \cite{Graham18}, whose main finding is that the two ecosystems are in fact connected. Nutrients, in particular nitrogen, leach from seabird guano onshore to nearshore marine systems through rainfall, among others. With this extra nutrient supply, benthic algae develop more on coral reefs adjacent to islands where seabirds are abundant, making the reef-fish communities there more abundant as well, given that fish mostly feed on those algae. 

\ppn The nitrogen input, as described in \cite{Graham18}, is thus a latent variable, positively associated to both seabird and fish densities. Now, by design (see Section \ref{subsec:puredep}), the Hellinger correlation should be particularly effective at highlighting dependence when induced by such a hidden effect. Indeed the empirical Hellinger correlation (\ref{eqn:empHellcorr}), with Legendre basis and $K$ and $L$ determined by cross-validation in (\ref{eqn:Btildestar}), is here $\hat{\eta}= 0.744$. The exact $p$-value of significance, obtained from $M=10,000$ independent bivariate uniform samples of size $n=12$, as described in Section \ref{sec:indeptest}, is found to be 0.013. So, at the typical significance level $\alpha = 0.05$, the Hellinger correlation is able to highlight the dependence between `seabirds' and `fish', even from such a small sample ($n=12$), corroborating \cite{Graham18}'s findings. 

\ppn One can also build a bootstrap confidence interval for $\eta$ by resampling from the empirical beta copula \citep{Segers17}. The conventional bootstrap (sampling the pairs with replacement from $\{\XX_1,\ldots,\XX_n \}$) is not appropriate here as the bootstrap resamples would typically include several times the same $\XX_i$. The associated $\hat{R}_i$ would then be 0, which would lead to gross underestimation of $\hat{\Bs}$ in (\ref{eqn:B12hat}). Indeed it is known that the conventional bootstrap is inconsistent in case of estimators being non-smooth functional of $F_{12}$, like estimators involving nearest-neighbor distances. On the other hand, sampling from the empirical beta copula can be regarded as a variant of the `smoothed' boostrap \citep{Kiriliouk19}. The empirical beta copula being continuous, the bootstrap resamples do not contain any repeated observation, leading to a consistent procedure. The variance of estimator (\ref{eqn:B12hat}) is not easily tractable, though \citep[p.\ 235]{Ebner18}. Hence we opt for the double boostrap procedure described in \cite{Karlsson00}. The produced two-sided bootstrap-$t$ confidence intervals are known to be second-order accurate \citep[Section 3.5]{Hall92}. For the current data set, the obtained 95\%-confidence interval for $\eta$ is $[0.67,1]$.

\section{Perspectives} \label{sec:persp}

The Hellinger correlation coefficient defined in this paper only applies to the case of univariate marginal distributions. A natural generalisation is to consider the case of multivariate marginals. Specifically, let $\XX_1 \in \R^{d_1}$ and $\XX_2 \in \R^{d_2}$, $d_1,d_2 > 1$, with continuous distributions $F_1$ and $F_2$ respectively. Let $\UU_1 = F_1(\XX_1)$ and $\UU_2=F_2(\XX_2)$ (componentwise), their copula transforms. Let $c_{12}$ be the $d_1+d_2$-dimensional copula density of $(\XX_1,\XX_2)$ on $\Is^{d_1+d_2}$ (joint density of $(\UU_1,\UU_2)$), $c_1$ be the $d_1$-dimensional copula density of $\XX_1$ on $\Is^{d_1}$ (joint density of $\UU_1$) and $c_2$ be the $d_2$-dimensional copula density of $\XX_2$ on $\Is^{d_2}$ (joint density of $\UU_2$). Then a marginal-free measure of dependence between the vectors $\XX_1$ and $\XX_2$ is the Hellinger distance between the copula $C_{12}$ and the product of the marginal copulas $C_1 C_2$, i.e., $\Hs^2 = \frac{1}{2} \iint_{\Is^{d_1+d_2}} \left(\sqrt{c}_{12}(\uu_1,\uu_2)  - \sqrt{c}_1(\uu_1) \sqrt{c}_2(\uu_2) \right)^2 \,d\uu_1 d\uu_2$. 
By Parseval identity, this is
\begin{equation}  \Hs^2 = \frac{1}{(2\pi)^{d_1+d_2}} \iint_{\R^{d_1+d_2}} \left| \varphi_{12}(\sss,\ttt) -\varphi_1(\sss) \varphi_2(\ttt)\right|^2 \,d\sss d\ttt \label{eqn:multH2} \end{equation}
where for $\sss \in \R^{d_1}$ and $\ttt \in \R^{d_2}$ and $i = \sqrt{-1}$, $\varphi_{12}(\sss,\ttt)  = \iint_{\Is^{d_1+d_2}} e^{i (\sss'\uu_1+\ttt'\uu_2)} \sqrt{c}_{12}(\uu_1,\uu_2) \,d\uu_1 d\uu_2$, $\varphi_{1}(\sss) = \iint_{\Is^{d_1}} e^{i \sss'\uu_1} \sqrt{c}_{1}(\uu_1) \,d\uu_1$ and $\varphi_{2}(\ttt) = \iint_{\Is^{d_2}} e^{i \ttt'\uu_2} \sqrt{c}_{2}(\uu_2) \,d\uu_2$,
% \begin{align*}
%  \varphi_{12}(\sss,\ttt) & = \iint_{\Is^{d_1+d_2}} e^{i (\sss'\uu_1+\ttt'\uu_2)} \sqrt{c}_{12}(\uu_1,\uu_2) \,d\uu_1 d\uu_2 \\
%   \varphi_{1}(\sss) & = \iint_{\Is^{d_1}} e^{i \sss'\uu_1} \sqrt{c}_{1}(\uu_1) \,d\uu_1 \\
%     \varphi_{2}(\ttt) & = \iint_{\Is^{d_2}} e^{i \ttt'\uu_2} \sqrt{c}_{2}(\uu_2) \,d\uu_2, 
% \end{align*}
that is, the Fourier transforms of $\sqrt{c}_{12}$, $\sqrt{c}_1$ and $\sqrt{c}_2$, respectively. Expanding the modulus in (\ref{eqn:multH2}), one obtains
\begin{equation} \Hs^2 = 1 - \frac{1}{(2\pi)^{d_1+d_2}}\Re\left\{\iint_{\mathbb{R}^{d_1+d_2}}\varphi_{12}(\sss,\ttt)\varphi^*_{1}(\sss)\varphi^*_{2}(\ttt)d\sss d\ttt\right\}, \label{eqn:multH22} \end{equation}
where $\cdot ^*$ denotes complex conjugation and $\Re(\cdot)$ is the real part. The three Fourier transforms $\varphi_{12}$, $\varphi_{1}$ and $\varphi_{2}$ are integrals with respect to the square-root of a copula density. As such, they can be estimated in a way very similar to estimating $\Bs$ by (\ref{eqn:B12hat}) or $\beta_{k\ell}$ by (\ref{eqn:tildebetakl}) in the 2-dimensional case, making use of nearest-neighbour distances in the copula-transformed domain.
Lemma 1 in \citet{Aya18} and Theorems 1 and 2 in \citet{Ebner18} would again guarantee the consistency of those estimators which, when plugged in (\ref{eqn:multH22}), would produce a consistent estimator of $\Hs^2$ for the multivariate-marginal case. This will be investigated in detail in a follow-up paper.

\ppn Another obvious question is how to adapt the Hellinger correlation to the case of discrete variables. In fact, (\ref{eqn:Hellfull}) obviously applies to the discrete case as well. If $(X_1,X_2)$ is a discrete random vector with joint probability mass function $p_{12}$ on $\Ss_1 \times \Ss_2$, where $\Ss_1$ and $\Ss_2$ are two discrete sets, and marginal distributions $p_1$ and $p_2$, respectively, then
\[ \Hs^2(X_1,X_2) = \frac{1}{2} \sum_{x_1 \in \Ss_1} \sum_{x_2 \in \Ss_2} \left(\sqrt{p}_{12}(x_1,x_2) - \sqrt{p}_1(x_1) \sqrt{p}_2(x_2) \right)^2. \]
However, this measure is not `margin-free'. Recently, \cite{Geenens19} proposed a concept of copula suitable for discrete random vectors, in which the marginals are not made uniform by Probability Integral Transform but by the Iterated Proportional Fitting procedure (IPF). One can then apply (\ref{eqn:Hellfull}) on $\overline{p}_{12}$, the `copula probability mass function', that is, the joint distribution after the marginals have been made discrete uniform by IPF, making it now margin-free. That measure will be the topic of a follow-up paper, as well as the related axiomatic. Indeed, it is clear for instance that the concept of `pure dependence' in (P6) only applies to continuous variables, and must be replaced by a different concept in the discrete case.

% \begin{figure}[h]
% 	\centering
% 	\includegraphics[width=0.55\textwidth]{rats.pdf}
% 	\caption{Fish density versus Nitrogen input for $n=12$ Chagos Islands: original data (left); pseudo-observations in the copula domain (right).}
% 	\label{fig:rats}
% \end{figure}

%\section{Discussion}

%-------------------------------------------------------------

%\section*{Acknowledgement}
%
%We are grateful to the editor, associate editor and referees for their insightful comments. This
%work was supported by the ARC.

%\section*{Supplementary Material}

%Supplementary material available at {\it Biometrika} online includes further details. 
%Estimation of the Hellinger correlation is implemented in the freely available R package \texttt{HellCor}.
\appendix
\section*{Appendix}

\section{Families of $\phi$-divergences} \label{App:phidiv}

Common choices for $\phi$ in (\ref{eqn:phidiv}) include:
\begin{enumerate}[$(i)$]
	\item $\phi_{L,\alpha}(t) = |t-1|^\alpha$, for $\alpha \geq 1$, yielding dependence measures $D_{L,\alpha}(X_1,X_2) \doteq  \Delta_{\phi_{L,\alpha}}(F_{12}||F_1F_2)$ reminiscent of some sort of $L_\alpha$-distance between $F_{12}$ and $F_1F_2$;
	\item $\phi_{P,\alpha}(t) = (t^\alpha - \alpha t +\alpha -1)/(\alpha(\alpha-1))$, for $\alpha \in \R \backslash \{0,1\}$, yielding dependence measures $D_{P,\alpha}(X_1,X_2) \doteq \Delta_{\phi_{P,\alpha}}(F_{12}||F_1F_2)$ reminiscent of the `power divergence' of \cite{Cressie84}. As $\lim_{\alpha \to 1} \phi_{P,\alpha}(t) = t\log t -t+1$, the Mutual Information $I(X_1,X_2) = \lim_{\alpha \to 1} D_{P,\alpha}(X_1,X_2)$ is a limiting case of this power-divergence family; %. It is also noted that R\'enyi $\alpha$-divergences \citep{Renyi61}, although not proper $\phi$-divergences, are in one-to-one correspondence with the corresponding power divergence, hence they are equivalent for most purposes;
	\item $\phi_{M,\alpha}(t) = |t^\alpha -1|^{1/\alpha}$ for $\alpha \in (0,1]$, yielding dependence measures $D_{M,\alpha}(X_1,X_2) \doteq \Delta_{\phi_{M,\alpha}}(F_{12}||F_1F_2)$, reminiscent of \cite{Matusita67}'s divergence.
\end{enumerate}
Now, in case $(i)$, it is seen that for $\alpha =1$, $\bar{\varphi} =2 < \infty$, hence the rescaled measure $D^\star_{L,1}(X_1,X_2) = (1/2) D_{L,1}(X_1,X_2)$ satisfies (P\ref{ax:exist})--(P\ref{ax:GDPI}). By contrast, if $\alpha > 1$, then $\bar{\varphi} = \infty$, and the corresponding measure is unable to characterise pure dependence. That is the case of Pearson's Mean Square Contingency coefficient $\Phi^2$, which corresponds to $\alpha = 2$. In case $(ii)$, one can check that, for $\alpha < 1$ (and $\alpha \neq 0$), $\bar{\varphi} = 1/(\alpha(1-\alpha)) < \infty$. Hence the associated dependence measure satisfies (P\ref{ax:exist})--(P\ref{ax:GDPI}). For $\alpha > 1$, $\bar{\varphi} = \infty$ and one faces the same issue as above. That includes the limiting case $\alpha \to 1$, which yields the Mutual Information. Finally, in case $(iii)$, $\bar{\varphi} =2 < \infty$ for all $\alpha \in (0,1]$, making any `Matusita' dependence measure $D^\star_{M,\alpha}(X_1,X_2)=(1/2)D_{M,\alpha}(X_1,X_2)$ comply with (P\ref{ax:exist})--(P\ref{ax:GDPI}). 

All in all, among the $\phi$-measures (\ref{eqn:phidiv}) of type $(i)$-$(ii)$-$(iii)$, only $D^\star_{L,1}(X_1,X_2)$, $D^\star_{P,\alpha}$ for $\alpha <1$ (and $\alpha \neq 0$) and $D^\star_{M,\alpha}$ for $0< \alpha \leq 1$ satisfy (P\ref{ax:exist})--(P\ref{ax:GDPI}). Evidently, $D^\star_{M,1}(X_1,X_2) = D^\star_{L,1}(X_1,X_2)$. This particular measure, forming a kind of intersection between the $L$- and $M$-families, is \cite{Silvey64}'s $\Delta$ described in Section \ref{sec:comphi}. At the intersection of the $P$- and $M$-families lies $D^\star_{P,1/2}(X_1,X_2) = D^\star_{M,1/2}(X_1,X_2)$, which is $\Hs^2$ given in (\ref{eqn:Hellfull}).

\section{Marginal transformations} \label{App:margtrans}

Lemma 1 in \citet{Aya18} and Theorems 1 and 2 in \citet{Ebner18} establish the root-$n$ consistency and the asymptotic normality of the oracle estimator $\widetilde{\Bs}$. However, those results hold only if the copula density $c_{12}$ is bounded and bounded away from 0 on $\Is^2$. This is a very restrictive assumption. In particular, many common parametric copula densities would grow unbounded in some of the corners of $\Is^2$ in the presence of dependence. Now, define a double marginal transformation $\TT_i = (T_{i1},T_{i2}) = (\Xi_1^{-1}(U_{1i}),\Xi_{2}^{-1}(U_{2i}))$ where for $k=1,2$, $\Xi_k^{-1}$ is the quantile function of a continuous distribution $\Xi_k$ having a bounded density $\xi_k$ on $\Is$. Standard developments show that $\{\TT_1,\ldots,\TT_n\}$ is a sample from a distribution having density $d_{12}(t_1,t_2) = c_{12}(\Xi_1(t_1),\Xi_2(t_2))\xi_1(t_1)\xi_2(t_2)$ on $(t_1,t_2) \in \Supp(\XXi) \doteq\Supp(\Xi_1) \times \Supp(\Xi_2)$. Then we see that
\begin{equation*} \Bs =  \iint_{\Is^2} \sqrt{c_{12}}(u_1,u_2)\,du_1\,du_2  = \iint_{\Supp(\XXi)} \sqrt{d_{12}}(t_1,t_2)\sqrt{\xi_1}(t_1) \sqrt{\xi_2}(t_2)\,dt_1\,dt_2 ,
\end{equation*}
by the obvious change-of-variable. Let $S_i = \min_{j \neq i} \|\TT_j-\TT_i\|_2$. Then, similarly to $\widetilde{\Bs}$, one can also estimate $\Bs$ by
\begin{equation} \widetilde{\widetilde{\Bs}\ } = \frac{2\sqrt{n-1}}{n}\sum_{i=1}^n S_i \sqrt{\xi_1}(T_{i1}) \sqrt{\xi_2}(T_{i2}). \label{eqn:Bhat} \end{equation}
The known weight function $\sqrt{\xi_1} \sqrt{\xi_2}$ is easily accounted for in the theoretical developments; see \citet[Lemma 1]{Aya18} or \citet[Theorem 1]{Ebner18}.

The transformations $\Xi_1$ and $\Xi_2$ can be whatever is convenient. In \cite{Geenens17}, their role was essentially to send the boundaries `far away' from the observations so as to annihilate boundary effects. Here one could take, for instance, $\Xi_1$ and $\Xi_2$ to be Beta$(6,6)$-distributions: the ensuing density $d_{12}$ would remain supported on $\Is^2$, but concentrated around the center of it due to its Beta$(6,6)$-marginals. That `double Beta transformation' is also beneficial in terms of relaxing the above mentioned assumption on $c_{12}$. Indeed, let $c_{12}$ satisfy Assumption 3.3 in \cite{Geenens17}, which is rather mild and allows $c_{12}$ to grow unboundedly in some of the corners of $\Is^2$. With $\xi_k(t_k) \propto t_k^5(1-t_k)^5$ and $\Xi_k(t_k) = \int_0^{t_k} \xi_k(t^*)\,dt^*$, for $k=1,2$, that is the Beta$(6,6)$ density and cumulative distribution functions, it can be shown that $d_{12}$ is H\"older continuous (with exponent $\alpha = 2$) on $\Is^2$ -- this follows as in Lemma A.1 in \cite{Geenens17}. 

Then, Corollary 8 of \cite{Singh16} applies, and $\widetilde{\widetilde{\Bs}\ }$ is root-$n$ consistent for $\Bs$. The results shown in Figures \ref{fig:BP} and \ref{fig:HHGplots} were actually obtained making use of the Beta$(6,6)$ transformation.

\section{Cross-validation} \label{App:CV}

For nonparametric functional estimation through orthogonal series approximation, it is well-known that the truncation cutoff plays the role of smoothing parameter \cite[Chapter 2]{Efromovich99}. Hence, if one wants to estimate $\sqrt{c}_{12}$ by 
\[\widehat{\sqrt{c}}_{12}(u_1,u_2)= \sum_{k=0}^K \sum_{\ell=0}^L \hat{\beta}_{k\ell} b_k(u_1)b_\ell(u_2), \]
as suggested in Section \ref{subsec:norm}, one should select $K$ and $L$ appropriately. The usual Mean Integrated Squared Error of that estimator is
\begin{align*} \iint_{\Is^2} & \left(\widehat{\sqrt{c}}_{12}(u_1,u_2) - \sqrt{c}_{12}(u_1,u_2)\right)^2  \,du_1 du_2  \\ & = 1 + \iint_{\Is^2} \left( \widehat{\sqrt{c}}_{12}\right)^2(u_1,u_2)\,du_1 du_2 - 2 \iint_{\Is^2} \widehat{\sqrt{c}}_{12}(u_1,u_2) \sqrt{c}_{12}(u_1,u_2)\, du_1 du_2 , \end{align*}
that we may seek to minimise with respect to $K$ and $L$. We know that
\[ A^2(K,L) \doteq \iint_{\Is^2} \left( \widehat{\sqrt{c}}_{12}\right)^2(u_1,u_2)\,du_1 du_2 = \sum_{k=0}^K \sum_{\ell=0}^L \hat{\beta}^2_{k\ell}, \]
obviously an increasing function in both $K$ and $L$. On the other hand, it follows from \citet[Lemma 1]{Aya18} that, for any bounded function $f: \Is^2 \to \R$, $\iint_{\Is^2} f(u_1,u_2) \sqrt{c}_{12}(u_1,u_2)\, du_1 du_2$ can be estimated by:
\[ \frac{2\sqrt{n-1}}{n}\sum_{i=1}^n \hat{R}_i f(\hat{U}_{i1},\hat{U}_{i2}); \]
compare (\ref{eqn:tildebetakl}). This justifies to estimate $\iint_{\Is^2} \widehat{\sqrt{c}}_{12}(u_1,u_2) \sqrt{c}_{12}(u_1,u_2)\, du_1 du_2$ by 
\[ B(K,L) \doteq\frac{2\sqrt{n-1}}{n}\sum_{i=1}^n \hat{R}_i \widehat{\sqrt{c}}^{(-i)}_{12}(\hat{U}_{i1},\hat{U}_{i2}) \]
where, as usual, the Leave-one-Out version of the estimator $\widehat{\sqrt{c}}^{(-i)}_{12}$ is used for avoiding overfitting. Explicitly, this is
\[\widehat{\sqrt{c}}^{(-i)}_{12}(u_1,u_2) = \sum_{k=0}^K \sum_{\ell=0}^L \hat{\beta}^{(-i)}_{k\ell} b_k(u_1)b_\ell(u_2)\]
where
\[ \hat{\beta}^{(-i)}_{k\ell} = \frac{2 \sqrt{n-2}}{n-1} \sum_{i'\neq i} \hat{R}^{(-i)}_{i'} b_k(\hat{U}_{i'1})b_\ell(\hat{U}_{i'2})\]
and $\hat{R}^{(-i)}_{i'} = \min_{j\notin \{i,i'\}} \|\widehat{\UU}_j-\widehat{\UU}_{i'}\|_2$. This yields the explicit expression
\begin{align*}
B(K,L) & = \frac{2\sqrt{n-1}}{n}\sum_{i=1}^n \hat{R}_i \sum_{k=0}^K \sum_{\ell=0}^L \hat{\beta}^{(-i)}_{k\ell} b_k(\hat{U}_{i1})b_\ell(\hat{U}_{i2}) \\
& = \frac{4 \sqrt{n-2}}{n\sqrt{n-1}} \sum_{k=0}^K \sum_{\ell=0}^L \sum_{i=1}^n \sum_{i'\neq i} \hat{R}_i \hat{R}^{(-i)}_{i'} b_k(\hat{U}_{i'1})b_\ell(\hat{U}_{i'2}) b_k(\hat{U}_{i1})b_\ell(\hat{U}_{i2}).
\end{align*}
Finally, $K$ and $L$ may be chosen as the values which minimise $A^2(K,L) - 2 B(K,L)$, which are easy to identify given that $K$ and $L$ are integers.

%\section{Empirical illustration - $n=5000$} \label{App:5000}
%
%%\setcounter{figure}{2}    
%
%
%\begin{figure}[H]
%	\centering
%	\includegraphics[width=0.99\textwidth]{HHG-5000.pdf}
%	\caption{Fifteen random samples of size $n=5000$ generated from the 15 scenarii of \citet[Figure 4]{Heller16}. The empirical Hellinger correlations are shown on top.}
%	\label{fig:HHGplots5000}
%\end{figure}

\end{document}